\documentclass[12pt,oneside]{article}

\usepackage[round,colon,authoryear]{natbib} 

\usepackage{graphicx}
\usepackage{amsfonts}
\usepackage{amssymb}
\usepackage{amsmath}
\usepackage{amsthm}
\usepackage{enumerate}
\usepackage{times}
\usepackage{eurosym}
\usepackage{url}

\newcommand{\R}{\mathbb{R}}
\newcommand{\E}{\mathbb{E}}
\newcommand{\N}{\mathbb{N}}
\newcommand{\Prob}{\mathbb{P}}
\newcommand{\Q}{\mathbb{Q}}

\newcommand{\1}{\mathbf{1}}
\newcommand{\CF}{\mathcal{F}}

\newcommand{\Qd}{\Q^{\$}}
\newcommand{\eu}{\text{\euro}}
\newcommand{\Qe}{\Q^{\eu}}
\newcommand{\dd}{\mathrm{d}}

\newcommand{\tInd}{_{t \in [0,T]}}

\newtheorem{thm}{Theorem}
\newtheorem{prop}{Proposition}
\newtheorem{lemma}{Lemma}
\newtheorem{cor}{Corollary}

\theoremstyle{remark}

\theoremstyle{remark}

\theoremstyle{remark}

\theoremstyle{remark}
\newtheorem{remark}{Remark}
\theoremstyle{remark}

\allowdisplaybreaks

\title{Why are quadratic normal volatility models analytically tractable?\thanks{We thank Alex Lipton and Nicolas Perkowski for
their helpful comments on an early version of this paper.  We are grateful to the two anonymous referees and the associate editor for their careful reading and very helpful suggestions, which improved this paper.  The views represented herein are the authors' own views and do not necessarily represent the views of Morgan Stanley or its affiliates and are not a product of Morgan Stanley research.
}}
\author{Peter Carr\thanks{New York University, Courant Institute, E-mail: pcarr@nyc.rr.com} \and
    Travis Fisher\thanks{E-mail: traviswfisher@gmail.com} \and
    Johannes Ruf\thanks{University of Oxford, Oxford-Man Institute of Quantitative
    Finance and Mathematical Institute,
    E-mail: johannes.ruf@oxford-man.ox.ac.uk}}

\begin{document}
\thispagestyle{plain} \maketitle

\begin{abstract}
\noindent
We discuss the class of ``Quadratic Normal Volatility" (QNV) models, which have drawn much attention in the financial industry due to their analytic tractability and flexibility. We characterize these models as those that can be obtained from stopped Brownian motion by a simple transformation and a change of measure that  depends only on the terminal value of the stopped Brownian motion. This explains the existence of explicit analytic formulas for option prices within QNV models in the academic literature.  Furthermore, via a different transformation, we connect
a certain class of QNV models to the dynamics of geometric Brownian motion and discuss
changes of num\'eraires if the num\'eraire is modelled as a QNV process.
\end{abstract}
\section{Introduction and model}  \label{SS qnv model}
Quadratic Normal Volatility (QNV) models  have recently drawn much attention in both industry and academia since they
are not only easily tractable as generalizations of the standard Black--Scholes framework
but also can be well calibrated to various market scenarios due to their flexibility.
In this paper, we focus on associating the dynamics of QNV processes with the dynamics of Brownian motion and geometric Brownian motion.   These relationships reveal why analytic formulas for option prices can be (and indeed have been) found. However, we shall abstain here from computing explicit option prices implied by a QNV model. Formulas for these can  be found, for example, in \citet{Andersen}.

It is well known that, in complete strict local martingale models, interpreting expectations as prices leads to seemingly paradoxical contingent claim prices. For example, standard put-call parity  is not satisfied in such models, as \citet{Andersen} discusses.  However, as we illustrate in our companion paper, \citet{CFR2011}, a simple adjustment to standard prices yields a pricing rule that bypasses those issues and, moreover, leads to prices that can be interpreted, in some sense, as minimal hedging costs.  With these considerations in mind, the reader should not worry about the fact that, under certain parameter constellations, a QNV process is not a true martingale, but a strict local martingale. Instead, the reader should keep in mind that one always can use sums of certain expectations as (adjusted) prices and thus avoid all those issues related to the pricing of contingent claims with strict local martingales as underlying.  In Section~\ref{S hyperinflation} of this paper, we provide formulas to compute these prices for arbitrary path-dependent contingent claims.

After introducing QNV models in this section and providing an overview of the relevant literature, we show in Section~\ref{SS Wiener} how QNV models can be obtained from transforming a stopped Brownian motion.  
In Section~\ref{S two real}, we work out a connection between a certain class of QNV processes and geometric Brownian motion, and, in 
Section~\ref{S closedness}, we formalize the observation that QNV models are stable under changes of num\'eraires.  Section~\ref{S semistatic} contains some preliminary results on semistatic hedging within QNV models, and
Section~\ref{S hyperinflation} provides an interpretation of the strict local martingale dynamics of certain QNV processes as the possibility of a hyperinflation under a dominating measure, in the spirit of \citet{CFR2011}. The appendix
contains a technical result.

\subsubsection*{Model}
If not specified otherwise, we 
work on a filtered probability space $(\Omega, \mathcal{F}, \{\CF_t\}_{t \geq 0}, \Q)$, equipped with a Brownian motion $B = \{B_t\}_{t \geq 0}$.  We introduce a process $Y= \{Y_t\}_{t \geq 0}$ with deterministic initial value $Y_0=y_0>0$, whose dynamics solve
\begin{align}   \label{E qnv X}
	\dd Y_t = (e_1 Y^2_t + e_2 Y_t + e_3) \dd B_t,
\end{align}
 where $e_1, e_2, e_3 \in \R$.  Problem~3.3.2 in \citet{McKean_1969} yields the
    existence of a unique, strong solution to this stochastic differential equation. 
We define $S$ as the first hitting time of zero by $Y$ and shall also study a stopped version $X = \{X_t\}_{t \geq 0}$ of $Y$, defined by $X_t := Y^S_t := Y_{t \wedge S}$ for all ${t \geq 0}$. 
We set $x_0 := X_0  = y_0$.

The dynamics of $Y$  and $X$ of course strongly depend on the parameters $e_1, e_2, e_3 $
 in the quadratic polynomial $P(z) := e_1 z^2 + e_2 z + e_3$ appearing in \eqref{E qnv X}.
 We shall say that $Y$ ($X$) is a (stopped) QNV process\footnote{We use the word ``normal'' in the name Quadratic Normal Volatility to emphasize the fact that we are interested in a model  where the normal local volatility is quadratic, as opposed to the lognormal local volatility, such as in the specification $(\dd \widetilde{Y}_t)/\widetilde{Y}_t = (e_1 \widetilde{Y}^2_t + e_2 \widetilde{Y}_t + e_3) \dd B_t$.} with polynomial $P$.
 The most important special cases are the following:
\begin{itemize}
    \item $e_1=e_2=0$, corresponding to Brownian motion;
    \item $e_1=e_3=0$, corresponding to geometric Brownian motion;
    \item $e_2=e_3=0$, corresponding to the reciprocal of a three-dimensional Bessel process.
\end{itemize}

Feller's test for explosions directly yields
that $Y$ does not hit any real roots of $P$ (except if $P(y_0)=0$,
in which case $Y \equiv y_0$ would just be constant); see Theorem~5.5.29 in \citet{KS1}.  The $\Q$-local martingales $Y$ and $X$ are not necessarily  true $\Q$-martingales. Indeed, 
the configuration of the roots of $P$ determines whether $Y$ and $X$ are true
$\Q$-martingales or strict $\Q$-local martingales: 
\begin{prop}[Martingality of QNV processes]   \label{P qnv martingality}
    The cases where the QNV process $Y$ is a true martingale are exactly the case when $e_1 = 0$ and the case when $P$ has two real roots $r_1, r_2 \in \R$ with $r_1 \leq r_2$ and $y_0 \in [r_1,r_2]$.  The cases where the   stopped QNV process $X$ is a true martingale are exactly the case when $e_1 = 0$ and the case when $P$ has a root $r \in \R$ with $x_0  \leq r$.
\end{prop}
The proposition is proved in Section~\ref{S closedness}.

\subsubsection*{Relevant literature}
An incomplete list of authors who study QNV
models in various degrees of generality consists of the following:
\begin{itemize}
    \item \citet{Rady_Sandmann}, \citet{Rady},
    \citet{Miltersen_1997}, \citet{Goldys_1997},  \citet{Ingersoll_bounded}, and \citet{Lipton_book}, who
    study the case when $X$ is bounded and thus a true martingale
    in the context of foreign exchange and interest rate markets; related is the class  of rational lognormal models introduced by Flesaker and Hughston \citep[see][]{FlesakerHughston, BrodyHughston};
    \item \citet{Albanese_BS_hyper} and \citet{Lipton_volsmile}, who derive prices for
    European calls on $X$ in the case
    when $P$ has one or two real roots;
    \item \citet{Zuehlsdorff_2001, Zuehlsdorff_2002}, \citet{Andersen}, and \citet{Chibane}, who compute
    prices for
    European calls and puts in the general case.
\end{itemize}
Most of these papers focus on deriving analytic expressions for the pricing of
European-style contingent claims. We refer the reader to \citet{Andersen} for the precise formulas
of European call and put prices.  In the following sections, we shall derive purely probabilistic methods to easily compute the price of any, possibly path-dependent, contingent claim.

\section{Connection to Wiener process}  \label{SS Wiener}
\citet{Bluman_transformation, Bluman_mapping},
\citet{CLM_reduction}, and \citet{Lipton_book} prove that the partial differential equations (PDEs) corresponding to the class of QNV models are the only parabolic PDEs that can be
reduced to the heat equation, via a certain set of transformations.  In this section, we derive a probabilistic
equivalent while, in particular, paying attention to the issues of strict local martingality.  More precisely, we shall see
that if one starts on a Wiener space equipped with a Brownian motion $W = \{W_t\}_{t \geq 0}$, and one is allowed 
\begin{enumerate}
	\item to stop $W$ at a stopping time $\tau$, yielding  $W^\tau =  \{W_t^\tau\}_{t \geq 0} :=  \{W_{t\wedge \tau}\}_{t \geq 0}$,
	\item to transform $W^\tau$ by a strictly increasing smooth function $f$, and
	\item to change the probability measure with a density process $Z = \{Z_t\}_{t \geq 0}$ of the form $Z_t = \widetilde{g}(t,W_t^\tau)$ for all $t \geq 0$ for some nonnegative measurable function $\widetilde{g}$,
\end{enumerate}
then $\{f(W^\tau_t)\}_{t \geq 0}$ is, under the new measure, a QNV process up to time $\tau$, given that it is a local martingale.

Our proof of Theorem~\ref{T QNV}, which shows this characterization of QNV processes, relies on the characterization of the solutions of three ordinary differential equations (ODEs).
The next lemma relates the solutions of these ODEs to each other:
\begin{lemma}[Three ODEs]  \label{L ODE}
    Fix $C,\mu_0,d,f_0 \in \R$ and $a, b \in [-\infty, \infty]$ with $a < 0 < b$ and let $\mu: (a, b) \rightarrow \R$  solve the ODE
    \begin{align}
        \mu^\prime(x) - \mu(x)^2 &= C,  \text{ }\quad \mu(0) = \mu_0. \label{E ODE0}
    \end{align}
	Then the functions $f, g:  (a, b) \rightarrow \R$,   defined by
	\begin{align}
	f(x) &:= d \int_0^x \exp\left(2\int_0^y \mu(z) \dd z\right) \dd y +  f_0, \nonumber\\
    	g(x) &:= \exp\left(-\int_0^x \mu(z) \dd z\right),   \label{E def g}
	\end{align}
	solve the ODEs
    \begin{align}   
        f^\prime(x) &= e_1 f(x)^2 + e_2 f(x) + e_3, \text{ } \quad f(0) = f_0,  \label{E ODEf}\\ 
        -g^{\prime\prime}(x) &= C g(x), \text{ } \quad g(0) = 1, \text{ } \quad g^\prime(0) = -\mu_0,  \label{E ODEg}
    \end{align}	
	respectively, for appropriate $e_1, e_2, e_3 \in \R$.
\end{lemma}
\begin{proof}
    The ODE in \eqref{E ODEg} can be checked easily. To show \eqref{E ODEf},  first consider the case $\mu_0^2  = -C$.  The uniqueness of solutions to \eqref{E ODE0} yields that $\mu \equiv \mu_0$; see, for example, Section~8.2 in \citet{Hirsch}. If $\mu_0 \neq 0$, then $f(x) = d(\exp(2 \mu_0 x) - 1)/(2\mu_0) + f_0$, and if $\mu_0=0$, then
$f(x) = dx + f_0$ both satisfy  \eqref{E ODEf}.  Consider now the case $\mu_0^2  \neq -C$ and observe that then $\mu(x)^2  \neq -C$ for all $x \in (a,b)$, again by a similar uniqueness argument. We obtain that
	\begin{align*}
        \log\left(\frac{\mu^\prime(x)}{\mu^2_0 + C}\right)
        = \log\left(\frac{\mu(x)^2+C}{\mu^2_0 + C}\right) =  \int_0^x \frac{2 \mu(z) \mu^\prime(z)}{\mu(z)^2 + C} \dd z = 2 \int_0^x \mu(z) \dd z.
    \end{align*}
	Therefore, $f(x) = d(\mu(x)  - \mu_0)/(\mu^2_0 + C)  + f_0$, which satisfies \eqref{E ODEf}.
\end{proof}

The next lemma provides the full set of solutions for the ODEs in \eqref{E ODEf} and \eqref{E ODEg}:
\begin{lemma}[Solutions of ODEs]
\label{L functionsODE}
	Fix $e_1, e_2, e_3, f_0 \in \R$ and set $\mu_0 =e_1 f_0 + e_2/2$,  $C = e_1 e_3 - e_2^2/4$, $r_1 = (-e_2/2 -\sqrt{|C|})/e_1$, and $r_2 = (-e_2/2 +\sqrt{|C|})/e_1$. Then the ODE in \eqref{E ODEf} has a unique solution $f$, which is defined in a neighborhood $(a,b)$ around zero, with $a,b \in [-\infty, \infty]$ and $a<0<b$, such that 
	\begin{itemize}
		\item $\lim_{x \downarrow a} |f(x)| = \infty$ if $a>-\infty$, 
		\item $\lim_{x \uparrow b} |f(x)| = \infty$ if $b < \infty$, 
		\item $f^\prime(x) \neq 0$ for all $x \in (a,b)$, and 
		\item $f$ is infinitely differentiable in $(a,b)$.
	\end{itemize}
Furthermore, the function $\mu: (a,b) \rightarrow \R$, defined by
    \begin{align} \label{E fMu}
        \mu(x) := \frac{1}{2} \frac{f^{\prime\prime}(x)}{f^{\prime}(x)} = e_1 f(x) + \frac{e_2}{2},
    \end{align}
	satisfies \eqref{E ODE0}.  The function $g: (a,b) \rightarrow  \R$, defined as in \eqref{E def g}, satisfies $\lim_{x \downarrow a} g(x) = 0$ if $a>-\infty$ and $\lim_{x \uparrow b} g(x) = 0$ if $b < \infty$;  thus, its domain can be extended to $[a,b] \cap \R$.
In explicit form, the functions $g$ and $f$  are as follows:
	\begin{itemize}
		\item if $e_1 = 0$:
			\begin{align*}
        g(x) &=  \exp\left(- \frac{e_2x}{2}\right),\\
        f(x) &=  \left(f_0 + \frac{e_3}{e_2}\right) \exp(e_2 x) - \frac{e_3}{e_2}\quad
            (\text{if } e_2 \neq 0) \quad \text { or } \quad f(x) = e_3 x +f_0 \quad (\text{if } e_2 = 0);
		\end{align*}
		\item if $C = 0$ and $e_1 \neq 0$:\footnote{This condition is equivalent to the condition that the polynomial $P$, defined by $P(z) = e_1 z^2 + e_2 z + e_3$, has a double root $r \in \R$. The quadratic formula then yields that $r = r_1=r_2$.}
			\begin{align*}
        g(x) &= 1-\mu_0 x = 1- \left(e_1 f_0 + \frac{e_2}{2}\right) x,\\
        f(x) &=   \left(f_0 - r_1\right) \frac{1}{1-\mu_0 x} + r_1 = \left(f_0 + \frac{e_2}{2e_1}\right) \frac{1}{1-\left(e_1 f_0 + \frac{e_2}{2}\right) x} - \frac{e_2}{2 e_1};
		\end{align*}
		\item  if $C < 0$ and $e_1 \neq 0$:\footnote{This condition is equivalent to the condition that the polynomial $P$, defined as above, has two real roots, which are exactly at $r_1, r_2$, by the quadratic formula.}
			\begin{itemize}
				\item and, additionally, $f_0 \in (r_1 \wedge r_2, r_1 \vee r_2)$:\footnote{This condition is equivalent to $\mu_0 \in (-\sqrt{-C}, \sqrt{-C})$.}
			\begin{align*}
        g(x) &= \frac{\cosh\left(\sqrt{-C}x + c\right)}{\cosh(c)} = \frac{\cosh\left(\sqrt{\frac{e_2^2}{4}- e_1 e_3}x + c\right)}{\cosh(c)},\\
        f(x) &=  \frac{-\sqrt{-C}}{e_1} \tanh\left(\sqrt{-C}x + c\right) - \frac{e_2}{2 e_1} \\&= \frac{-\sqrt{\frac{e_2^2}{4} - e_1 e_3}}{e_1} \tanh\left(\sqrt{\frac{e_2^2}{4} - e_1 e_3}x + c\right) - \frac{e_2}{2 e_1}
		\end{align*}
				with 
			\begin{align*}
				c := \mathrm{arctanh}\left(\frac{-\mu_0}{\sqrt{-C}}  \right) = \frac{1}{2} \log\left(\frac{\sqrt{-C} - \mu_0}{\sqrt{-C} + \mu_0}\right);  
			\end{align*}
				\item and, additionally, $f_0 \notin [r_1 \wedge r_2, r_1 \vee r_2]$:\footnote{This condition is equivalent to $\mu_0 \notin [-\sqrt{-C}, \sqrt{-C}]$.}
			\begin{align*}
        g(x) &= \frac{\sinh\left(\sqrt{-C}x + c\right)}{\sinh(c)} = \frac{\sinh\left(\sqrt{\frac{e_2^2}{4} - e_1 e_3}x + c\right)}{\sinh(c)},\\
        f(x) &=  \frac{-\sqrt{-C}}{e_1} \coth\left(\sqrt{-C}x + c\right) - \frac{e_2}{2 e_1} \\&= \frac{-\sqrt{\frac{e_2^2}{4} - e_1 e_3}}{e_1} \coth\left(\sqrt{\frac{e_2^2}{4}- e_1 e_3}x + c\right) - \frac{e_2}{2 e_1}
		\end{align*}
				with 
			\begin{align*}
				 c := \mathrm{arccoth}\left(\frac{-\mu_0}{\sqrt{-C}}  \right) = \frac{1}{2} \log\left(\frac{-\sqrt{-C} + \mu_0}{\sqrt{-C} + \mu_0}\right)  ;
			\end{align*}
				\item and, additionally, $f_0 \in \{r_1 \wedge r_2, r_1 \vee r_2\}$:\footnote{This condition is equivalent to $\mu_0 \in \{-\sqrt{-C}, \sqrt{-C}\}$.}
			\begin{align*}
        g(x) &=\exp(-\mu_0 x),\\
        f(x) &\equiv f_0;
		\end{align*}
		\end{itemize}
\item if $C > 0$:\footnote{This condition is equivalent to the condition that the polynomial $P$, defined as above, has no real roots.}
			\begin{align*}
        g(x) &= \frac{\cos\left(\sqrt{C}x + c\right)}{\cos(c)} = \frac{\cos\left(\sqrt{e_1 e_3 - \frac{e_2^2}{4} }x + c\right)}{\cos(c)},\\
        f(x) &=  \frac{\sqrt{C}}{e_1} \tan\left(\sqrt{C}x + c\right) - \frac{e_2}{2 e_1} = \frac{\sqrt{e_1 e_3 - \frac{e_2^2}{4} }}{e_1} \tan\left(\sqrt{e_1 e_3 - \frac{e_2^2}{4}}x + c\right) - \frac{e_2}{2 e_1}
		\end{align*}
				with 
			$
				c := \arctan({\mu_0}/{\sqrt{C}} ).  
			$
	\end{itemize}
\end{lemma}
\begin{proof}
    The existence and uniqueness of solutions of the ODE in \eqref{E ODEf} follow  from standard arguments in the
    theory of ODEs; see Section~8.2 in \citet{Hirsch}. The fact that $\mu$ satisfies the corresponding ODE and the other statements can be  checked easily.
\end{proof}

The description of the solutions of the ODEs in \eqref{E ODEf} and \eqref{E ODEg} is
the fundamental step in proving the following theorem, which characterizes QNV processes as
the only local martingales that can be simulated from stopped Brownian motion by a certain set of transformations:
\begin{thm}[QNV process and Brownian motion]  \label{T QNV}
 	Let $\Omega = C([0,\infty), \R)$ be the set of  continuous paths $\omega: [0, \infty) \rightarrow \R$.
	Let $\mathbb{F} = \{{\CF}_t\}_{t \geq 0}$ denote the filtration generated by the canonical process $W$, defined by $W_t(\omega) := \omega(t)$ for all $t \geq 0$,  and set ${\CF} = \bigvee_{t \geq 0} {\CF}_t$. Let ${\Prob}$ denote the Wiener measure on $({\Omega}, {\CF})$. 
Then the following statements hold:
	\begin{enumerate}
		\item    Fix $a,b \in [-\infty, \infty]$ with $a < 0<b$ and let $f: [a,b] \rightarrow [-\infty, \infty]$ and $\widetilde{g}: [0,\infty) \times [a,b] \rightarrow [0,\infty)$ denote
    two measurable functions with  $\widetilde{g}(0,0) = 1$, $\lim_{x \downarrow a} |f(x)| = \infty = |f(a)|$ if $a>-\infty$, and $\lim_{x \uparrow b} |f(x)| = \infty = |f(b)|$ if $b<\infty$. Assume that $f$ is three times continuously differentiable in $(a,b)$ and that $f^\prime(x) \neq 0$ for all $x\in (a,b)$.
 Define the $\mathbb{F}$-stopping time $\tau$ by
    \begin{align}  \label{E t1 tau}
        \tau := \inf \{t \geq 0| W_t \notin (a,b)\}, \text{ } \inf \emptyset := \infty,
    \end{align}
    and the processes $Y = \{Y_t\}_{t \geq 0}$ and $Z = \{Z_t\}_{t \geq 0}$ by $Y_t := f(W_t^\tau)$ and $Z_t := \widetilde{g}(t, W^\tau_t)$ for all $t \geq 0$. Assume that $Z$ is a (nonnegative) 
	$\Prob$-martingale and that $Y$ is a $\Q$-local martingale, where $\Q$ denotes the unique probability measure on $(\Omega, \bigvee_{t \geq 0} \CF_t)$ satisfying $\dd {\Q} / \dd \Prob|_{\CF_t} = Z_t$ for all $t \geq 0$.\footnote{See page~192 in \citet{KS1}.}

    Then $Y$ under $\Q$ satisfies
    \begin{align}  \label{E QNV UBracket}
        \dd Y_t = (e_1 Y^2_t + e_2 Y_t + e_3) \dd B_t
    \end{align}
    for all  $t \geq 0$ and for some $e_1,e_2,e_3\in \R$  and a $\Q$-Brownian motion $B = \{B_t\}_{t \geq 0}$; to wit, $Y$ is a $\Q$-QNV process.
    Furthermore, the corresponding density process $Z$ is of the form $Z_t =  \exp(C (t \wedge \tau) /2) g(W_t^\tau) = \exp(C t/2) g(W_t^\tau)$ for all $t \geq 0$, where 
	 $g$ is explicitly computed in Lemma~\ref{L functionsODE} (with the corresponding constants $e_1,e_2,e_3$ and $f_0 = f(0)$) and $C = e_1 e_3 - e_2^2/4$.	
	\item Conversely, for any $e_1,e_2,e_3, f_0 \in \R$
     there exist $a, b \in [-\infty, \infty]$ with $a < 0 < b$ and measurable functions $f: [a,b] \cap \R \rightarrow [-\infty, \infty]$ and $g: [a,b] \cap \R \rightarrow [0,\infty)$ such that the following hold:
		\begin{itemize}
			\item $f$ is infinitely differentiable in 
				$(a, b)$, $f^\prime(x) \neq 0$ for all $x \in (a,b)$, $\lim_{x \downarrow a} |f(x)| = \infty$ if $a>-\infty$, and $\lim_{x \uparrow b} |f(x)| = \infty$ if $b<\infty$;
			\item the process $Z = \{Z_t\}_{t \geq 0}$, defined by $Z_t = \exp(Ct/2) g(W_t^\tau)$ for all $t \geq 0$ with $C = e_1 e_3 - e_2^2/4$ and $\tau$ as in \eqref{E t1 tau}, is a $\Prob$-martingale and generates a probability measure $\Q$ on  $(\Omega, \bigvee_{t \geq 0} \CF_t)$;
			\item the process $Y = \{Y_t\}_{t \geq 0}$, defined by $Y_t = f(W_t^\tau)$ for all $t \geq 0$, has under $\Q$ the dynamics  in \eqref{E QNV UBracket}, for some $\Q$-Brownian motion $B = \{B_t\}_{t \geq 0}$, and satisfies $Y_0 = f_0$.
		\end{itemize}
	\end{enumerate}
\end{thm}
\begin{proof}
	We start with the proof of the first statement. 
	Observe that zero is an absorbing point of
    the martingale $Z$ since it is also a nonnegative supermartingale. Thus, by the martingale representation theorem \citep[see, for
    example, Theorem~III.4.33 of][]{JacodS}, there
    exists some progressively measurable process $\tilde{\mu} = \{\tilde{\mu_t}\}_{t \geq 0}$
    such that the dynamics of $Z$ can be described as $\dd Z_t
    = -  Z_t \tilde{\mu}_t \dd W_t$
    for all $t \geq 0$. 	Then Lenglart's extension of Girsanov's theorem implies that 
    the process $B = \{B_t\}_{t \geq 0}$, defined by
    $B_t := W_t + \int_0^t\tilde{\mu}_s \dd s$ for all $t \geq 0$, is a $\Q$-Brownian motion; see Theorem~VIII.1.12 in \citet{RY}.

    It\^o's formula yields that
    \begin{align} \label{E QNVIto}
        \dd Y_t = \dd f(W_t) = f^\prime(W_t)
            \dd W_t+ \frac{1}{2}f^{\prime\prime}(W_t) \dd t
            = f^\prime(W_t)
            \dd B_t+ \left(\frac{1}{2}f^{\prime\prime}(W_t)
            - f^\prime(W_t) \tilde{\mu}_t\right) \dd t
    \end{align}
    for all $t < \tau$.
    Now, the uniqueness of the Doob--Meyer decomposition
    \citep[see Theorem III.16 in][]{Protter}, in conjunction with \eqref{E QNVIto} and the fact that $Y$ is a $\Q$-local martingale, implies that
    $\tilde{\mu}_t = \mu(W_t)$ for all $t < \tau$, where we have defined
    \begin{align*}
        \mu(x) := \frac{1}{2} \frac{f^{\prime\prime}(x)}{f^{\prime}(x)}
    \end{align*}
    for all $x \in (a,b)$.
    By the assumption on $f$, the function $\mu$ is continuously differentiable in $(a,b)$.

    Then another application of It\^o's formula yields that
    \begin{align}  \label{E QNV g}
        \log(\widetilde{g}(t,W_t)) =  - \int_0^t \mu(W_s) \dd W_s - \frac {1}{2}
        \int_0^t \mu(W_s)^2 \dd s = -\eta(W_t) + \frac{1}{2}
        \int_0^t \left(\mu^\prime(W_s) - \mu(W_s)^2 \right) \dd s
    \end{align}
    for all $t < \tau$, where $\eta: (a,b)
    \rightarrow (-\infty, \infty)$ is defined as $\eta(x) =
    \int_0^x \mu(y) \dd y$ for all $x \in (a,b)$.
 Then
    Lemma~\ref{L pathindepence} in Appendix~\ref{A technical} yields
    that $\mu$ satisfies the ODE of \eqref{E ODE0} in $(a,b)$ for some $C \in \R$. 
	Lemma~\ref{L ODE} implies that $f$ solves the ODE of \eqref{E ODEf}; in conjunction with \eqref{E QNVIto}, this yields \eqref{E QNV UBracket}. Lemma~\ref{L functionsODE} yields that $C = e_1 e_3 - e_2^2/4$.
 The expression for $\widetilde{g}$ in \eqref{E QNV g} implies 
 that $Z$ is of the
    claimed form, first for all $t < \tau$ and then for $t = \tau$ since $Z_\tau = 0$ if $\tau < \infty$ because $\Q(Y_t \in \R \text{ for all } t \geq 0) = 1$. 

    For the converse direction, fix $e_1, e_2, e_3, f_0$ and apply Lemma~\ref{L functionsODE}. Define $Z$ as in the statement with the corresponding  function $g$, computed in Lemma~\ref{L functionsODE}.
	Assume for a moment that the corresponding
    change of measure exists, that is, $Z$ is a $\Prob$-martingale.  The computations in \eqref{E QNVIto}  then imply that $Y$ is a $\Q$-local martingale  with dynamics  as in \eqref{E QNV UBracket}, where $\Q$ is as in the statement.

	Thus, it remains to show that $Z$ is a $\Prob$-martingale.
This is clear in the cases $e_1 = 0$ and $C\geq 0$ since $Z$ then is either a stopped $\Prob$-(geometric) Brownian motion or a bounded $\Prob$-martingale by It\^o's formula. We conclude by writing $Z$ as the sum of two true martingales, in the case $C<0$, using $\sinh(x) := (\exp(x) - \exp(-x))/2$
    and $\cosh(x) := (\exp(x) + \exp(-x))/2$.
\end{proof}

The next corollary concludes this discussion by illustrating how expectations of path-dependent
functionals in QNV models can be computed. Here and in the following, we shall always assume $\infty \cdot \1_A(\omega) = 0$ if $\omega \notin A$ for any $A \in \CF$, for the sake of notation.
\begin{cor}[Computation of expectations in QNV models]  \label{C computation}
    Fix $T>0$, $y_0 > 0$, $e_1, e_2, e_3 \in \R$, and $C = e_1 e_3 - e_2^2/4$  and let $h:C([0,T],\R) \rightarrow [0,\infty]$ denote any nonnegative measurable function of continuous paths. 
	Let  $Y$ ($X$) denote a (stopped) QNV process with polynomial $P(z) = e_1 z^2 + e_2 z + e_3$ and $Y_0 = y_0$ ($=X_0$) and $W = \{W_t\}_{t\geq 0}$ a Brownian motion.

	Then there exist functions $f,g$ and a stopping time $\tau$ (adapted to the filtration generated by $W$) such that
    \begin{align*}
        \E[h(\{Y_t\}\tInd)] &= \E\left[h\left(\{f(W_t^\tau)\}\tInd\right) \1_{\{\tau > T\}} \exp\left(\frac{CT}{2}\right) g\left(W_T^\tau\right)\right],\\
        \E[h(\{X_t\}\tInd)] &= \E\left[h\left(\{f(W_t^{\tau \wedge S})\}\tInd\right) \1_{\{\tau > T \wedge S\}} \exp\left(\frac{C(T\wedge S)}{2}\right) g\left(W_T^{\tau \wedge S}\right)\right],
    \end{align*}
    where $S$ denotes the first hitting time of zero by the process $\{f(W_t)\}_{t \geq 0}$.
\end{cor}
\begin{proof}
    The statement follows directly from Theorem~\ref{T QNV}, first for bounded $h$ and then
    for any nonnegative $h$, by taking the limit after observing that, in the notation of Theorem~\ref{T QNV},
    $\{\tau > T\} = \{g(W_T^\tau) > 0\}$.
\end{proof}
With the notation of the last corollary and Lemma~\ref{L functionsODE}, let us define $\underline{W}_T := \min_{t \in [0,T]} W_t$ and $\overline{W}_T := \max_{t \in [0,T]} W_t$. Then
the event $E_1 := \{\tau > T\}$  can be represented as $\{a < \underline{W}_T <\overline{W}_T < b\}$. Without loss of generality assuming $\mu_0 \geq 0$,  we have the following:
\begin{itemize}
	\item  $E_1 = \Omega$ if $e_1 = 0$ or both $C<0$ and $\mu_0 \in [-\sqrt{-C}, \sqrt{-C}]$;
	\item  $E_1 = \{\overline{W}_T < 1/\mu_0\}$ if $e_1 \neq 0$ and $C=0$;
	\item $E_1 = \{\overline{W}_T < -c/\sqrt{-C}\}$ if $e_1 \neq 0$ and $C<0$ and $\mu_0 \notin
[-\sqrt{-C}, \sqrt{-C}]$;
	\item $E_1 = \{ (c-\pi/2)/\sqrt{C}<\underline{W}_T  <\overline{W}_T \leq (c+\pi/2)/\sqrt{C}\}$ if $C>0$.
\end{itemize}
The event $E_2 := \{\tau > T \wedge S\}$ has the same representation with $W$ always replaced by $W^S$. It can easily be checked that $E_2 = \Omega$ if the polynomial corresponding to the stopped QNV process $X$ has a root greater than $x_0$.
These considerations illustrate that for any QNV model, quantities of the form $\E[\tilde{h}(Y_T)]$ or $\E[\tilde{h}(X_T)]$ can easily be computed by using only the joint distribution of Brownian motion together with its running minimum and maximum.

For later use, we state the following simple observation:
\begin{lemma}[Reciprocal of a solution $f$]  \label{L reciprocal}
	Fix $e_1, e_2, e_3, f_0 \in \R$ with $f_0 \neq 0$.
	For the solution $f: (a,b) \rightarrow \R$ of \eqref{E ODEf}, with $a,b \in [-\infty, \infty]$ denoting possible times of explosion of $f$, consider $\mu$ and $g$, as in \eqref{E fMu} and \eqref{E def g}. Define $\widehat{a} = \sup\{x \in (a,0] | f(x) = 0\} \vee a$ and  $\widehat{b} = \inf\{x \in [0,b] | f(x) = 0\} \wedge b$ and consider the functions  $\widehat{f}, \widehat{\mu}, \widehat{g}: (\widehat{a}, \widehat{b}) \rightarrow \R$, defined by $\widehat{f}(x) := 1/ f(x)$, $\widehat{\mu}(x) := \widehat{f}^{\prime \prime}(x) / (2 \widehat{f}^{\prime}(x)) $, and $\widehat{g}(x) := \exp\left(-\int_0^x \widehat{\mu}(z) \dd z\right)$.

Then $\widehat{f}$ solves the Riccati equation
    \begin{align}  \label{E hat f}   
        \widehat{f}^\prime(x) &= -e_3 \widehat{f}(x)^2 - e_2 \widehat{f}(x) - e_1, \text{ } \quad \widehat{f}(0) = f_0^{-1},
	\end{align}
	and $\widehat{g}$ satisfies $\widehat{g} = g f/f_0$. In particular, if $\widehat{a} > -\infty$ (resp., $\widehat{b} < \infty$),
then $\lim_{x \downarrow \widehat{a}} g(x) f(x)$  (resp., $\lim_{x \uparrow \widehat{b}} g(x) f(x)$) exists and is real.
\end{lemma}
\begin{proof}
	Observe that $\widehat{f}^\prime(x) = - f^\prime(x)/f(x)^2$, which directly yields \eqref{E hat f}.
	Now, the identity $\widehat{\mu} = \mu - f^\prime/f$ implies that
	\begin{align*}
		\widehat{g}(x) &= \exp\left(-\int_0^x \widehat{\mu}(z) \dd z\right) = g(x) \exp\left(\int_0^x \frac{f^\prime(z)}{f(z)} \dd z\right) = \frac{g(x) f(x)}{f_0}.
	\end{align*}
	The existence of the limits follows as in Lemma~\ref{L functionsODE}.
\end{proof}

\section{Connection to geometric Brownian motion}  \label{S two real}
We now focus on the case when $Y$ is a QNV
process with a polynomial $P$ that has exactly two real roots $r_1, r_2$.
For that, we parameterize $P$ as $P(z) =
e_1 (z-r_1) (z -r_2)$ for some $e_1 \in \R\setminus\{0\}$
and $r_1, r_2 \in \R$ with $r_1 < r_2$. 

In the following, we shall connect the dynamics of a QNV process
to geometric Brownian motion. This link has been established for
the case $y_0 \in (r_1, r_2)$ in \citet{Rady}.
\begin{thm}[QNV process and geometric Brownian motion]  \label{P qnv GBM}
   Fix $T>0$, $e_1, r_1, r_2, y_0 \in \R$ with $e_1 \neq 0$ and $r_1<r_2$  and let $h:C([0,T],\R) \rightarrow [0,\infty]$ denote any nonnegative measurable function of continuous paths. 
   Let $Y = \{Y_t\}_{t \geq 0}$ denote a QNV process with polynomial $P(z) =
e_1 (z-r_1) (z-r_2)$  and $Y_0 = y_0$, and $Z= \{Z_t\}_{t \geq 0}$  a (possibly negative) geometric Brownian motion with $Z_0 = (y_0-r_2)/(y_0-r_1)$ and
    \begin{align}  \label{E dZstopped}
        \frac{\dd Z_t}{Z_t} = e_1 (r_2 - r_1) \dd B_t
    \end{align}
    for all $t \geq 0$, where $B = \{B_t\}_{t \geq 0}$
    denotes some Brownian motion. Let $\tau$ denote the first hitting time of $1$ by $Z$. Then we have that
    \begin{align}   \label{E X GBM}
\E[h(\{Y_t\}\tInd)] =
       \frac{y_0 - r_1}{r_2-r_1}
                \E\left[h\left(\left\{
                \frac{r_2-r_1 Z_t}{1 - Z_t}\right\}{\tInd}\right)  \1_{\{\tau > T\}}
                \left(1 - Z_T\right)\right].
    \end{align}
\end{thm}
\begin{proof}
	Define two processes $M = \{M_t\}_{t \geq 0}$ and $N = \{N_t\}_{t \geq 0}$ by
		\begin{align*}
			M_t  :=\frac{y_0 - r_1}{r_2-r_1} \cdot (1 - Z_t) \1_{\{\tau > t\}}\quad \text{ and } \quad N_t := \frac{r_2-r_1 Z_t}{1 - Z_t} \cdot \1_{\{\tau > t\}}
		\end{align*}
	for all $t \geq 0$. Then $M$ is a nonnegative martingale, started in one, and thus defines a new probability measure $\widetilde{\Q}$ by  $\dd \widetilde{\Q} = M_T \dd \Q$. Moreover, the right-hand side of \eqref{E X GBM} can be written as
    \begin{align*}   
                \E\left[h\left(\left\{N_t\right\}{\tInd}\right)  \1_{\{\tau > T\}} M_T\right] = \E^{\widetilde{\Q}}\left[h\left(\left\{N_t\right\}{\tInd}\right) \right]
    \end{align*}
	since $\widetilde{\Q}(\tau > T) = 1$.  
 	Therefore, and because of  $N_0 = y_0$, it is sufficient to show that $N$  is a $\widetilde{\Q}$-QNV process with polynomial $P$.   Observe that, under $\widetilde{\Q}$,  It\^o's formula yields that
	\begin{align*}
		\dd \langle N \rangle_t = \left(\frac{r_2 - r_1}{(1-Z_t)^2}\right)^2 \dd \langle Z\rangle_t = e_1^2 \frac{(r_2 - r_1)^4 Z_t^2}{(1-Z_t)^4} \dd t = e_1^2 (N_t - r_1)^2 (N_t - r_2)^2 \dd t = P(N_t)^2 \dd t.
	\end{align*}
	Thus, provided that $N$ is a continuous $\widetilde{\Q}$-local martingale on $[0,T]$, the one-dimensional version of Theorem~3.4.2 in \citet{KS1} (in conjunction with the weak uniqueness of solutions to the underlying stochastic differential equation) yields that $N$ is a QNV process with polynomial $P$ under $\widetilde{\Q}$.

	It remains to be shown that $N$ is a $\widetilde{\Q}$-local martingale on  $[0,T]$. Towards this end, since $\widetilde{\Q}(\tau>T) = 1$, it is sufficient to show that $N^{\tau_i}$ is a $\widetilde{\Q}$-martingale for all $i \in \N$, where $\tau_i$ is defined as the first hitting time of $1/i$ by $M$.
	However, this follows from the observation that $N^{\tau_i} M^{\tau_i}$ is a $\Q$-martingale and from Girsanov's theorem; see also Exercise~VIII.1.20 in \citet{RY}.
\end{proof}

In the setup of Theorem~\ref{P qnv GBM}, we observe that the process $Z$ is negative if and only if  $y_0 \in (r_1, r_2)$, exactly the case treated by  \citet{Rady}. In that case,  $\tau = \infty$ $\Q$-almost surely as $Z$  never hits $1$. It is also exactly this case when $Y$ is 
a true martingale; compare Proposition~\ref{P qnv martingality}.  Indeed, using $h(\omega) = \omega_T$ in \eqref{E X GBM} shows that 
\begin{align*}
	\E[Y_T] &=        \frac{y_0 - r_1}{r_2-r_1}
                \E\left[(r_2-r_1 Z_T) \1_{\{\tau > T\}}\right]  = \frac{y_0 - r_1}{r_2-r_1}
                \left(\E\left[r_2-r_1 Z_T^\tau \right] - (r_2 - r_1) \Q(\tau \leq T)\right)\\&= y_0 - (y_0 - r_1) \Q(\tau\leq T),
\end{align*}
 which equals $y_0$ if and only if $y_0 \in (r_1, r_2)$.

\begin{remark}[Brownian motion and three-dimensional Bessel process]
	With the notation of Theorem~\ref{P qnv GBM}, set $r_1 = 0$ and $y_0 = e_1 = 1$. Furthermore, let $\widetilde{M} = \{\widetilde{M}_t\}_{t \geq 0}$ denote a Brownian motion starting in one and stopped in zero
 and let $\widetilde{Y} = \{\widetilde{Y}_t\}_{t \geq 0}$ denote the reciprocal of a
    three-dimensional Bessel process starting in one; that is, $\widetilde{Y}$ is a QNV process with polynomial $\widetilde{P}(z) = z^2$ and satisfies $\widetilde{Y}_0 = 1$.  

	 Formally, as $r_2 \downarrow 0$, the dynamics of the QNV process $Y$, described by the polynomial $P(z) = z(z-r_2)$, resemble more and more those of $\widetilde{Y}$.
	Similarly, consider the process $M = \{M_t\}_{t \geq 0}$, defined in the proof of Theorem~\ref{P qnv GBM} as $M_t = (1-Z^\tau_t)/r_2$ for all $t \geq 0$, and observe that 
	\begin{align*}
        M_t = \frac{1 - Z_t^\tau}{r_2} = \frac{1}{r_2} \left(1 - (1-r_2)\exp\left(r_2 B_t^\tau - \frac{r_2^2 t}{2}\right)\right) = 1- B_t^\tau + O(r_2)
    \end{align*}
	for all $t \geq 0$ by a Taylor series expansion.  Thus, as $r_2 \downarrow 0$, in distribution the martingale $M$ resembles, more and more, the Brownian motion $\widetilde{M}$.

	Observe, furthermore,  that the process $N = \{N_t\}_{t \geq 0}$, defined in the proof of Theorem~\ref{P qnv GBM}, satisfies $N_t = 1/M_t$ for all $t \geq 0$.  Thus, Theorem~\ref{P qnv GBM} states, with the given parameters, that $N = 1/M$ has the same distribution as $Y$ after changing the probability measure with the Radon--Nikodym derivative $M_T$.  Indeed, it is well known that the reciprocal $\widetilde{N} = 1/\widetilde{M}$ has the dynamics of $\widetilde{Y}$ after changing the probability measure with the Radon--Nikodym derivative $\widetilde{M}_T$; see also \citet{Perkowski_Ruf}.  Thus, Theorem~\ref{P qnv GBM} extends this well-known relationship of Brownian motion and Bessel process to processes that are, in some sense, approximately Brownian motion and Bessel process.
\qed
\end{remark}

\section{Closedness under changes of measure}   \label{S closedness}

Before studying specific changes of measure involving QNV processes, let us make some general observations. Towards this end, fix a nonnegative continuous local martingale $\widetilde{X} = \{\widetilde{X}_t\}_{t \geq 0}$ with $\widetilde{X}_0 = 1$, defined on a filtered probability space $({\Omega},  \bigvee_{t \geq 0} \CF_t, \{\CF_t\}_{t \geq 0}, {\Prob})$.
Let $\{\tau_n\}_{n \in \N}$ denote the first hitting times of levels $n \in \N$ by $\widetilde{X}$ and observe that $\widetilde{X}$ defines a sequence of consistent probability measures $\{\widetilde{\Prob}_n\}_{n \in \N}$ on $\{\CF_{\tau_n}\}_{n \in \N}$  via $\dd \widetilde{\Prob}_n = (\lim_{t \uparrow \infty} \widetilde{X}^{\tau_n}_t) \dd \Prob$.  Under sufficient technical assumptions, this sequence of probability measures can be extended to a measure $\widetilde{\Prob}$ on $\bigvee_{n \in \N} \CF_{\tau_n}$; in particular, the probability space must be large enough to allow for an event that has zero probability under $\Prob$ but positive probability under $\widetilde{\Prob}$; for details and further references we refer the reader to Subsection~2.2 in \citet{Ruf_Novikov}. 

We remark that we usually may assume that the necessary technical assumptions hold by embedding $\widetilde{X}$ in a ``sufficiently nice'' canonical space, as long as we are interested in distributional properties of functionals of the path of $X$; see also Remark~1 in \citet{Ruf_Novikov}. We shall call $\widetilde{\Prob}$ the F\"ollmer measure\footnote{\citet{F1972} realized the usefulness of relating probability measures to nonnegative supermartingales, in particular, to nonnegative local martingales. \citet{M}, giving credit to H.~F\"ollmer, simplified the construction of these measures, and \citet{DS_Bessel} used these results to develop an understanding of strict local martingales in the theory of no-arbitrage.} corresponding to $\widetilde{X}$.

One can prove that
    \begin{align} \label{E Qe}
        \E^{\widetilde{\Prob}}\left[\frac{1}{\widetilde{X}_T} \left(H \1_{\{\widetilde{X}_T < \infty\}}\right)\right] =\E^{\Prob}\left[H \1_{\{\widetilde{X}_T>0\}}\right]
    \end{align}
for all $\CF_T$-measurable random variables $H \in [0,\infty]$ and $T \geq 0$, where we again have set $\infty \cdot \1_A(\omega) = 0$ if $\omega \notin A$ for any $A \in \CF_T$.
Furthermore, $1/\widetilde{X}$ is a $\widetilde{\Prob}$-local martingale and, moreover, $\widetilde{X}$ is a strict $\Prob$-local martingale if and only if $\widetilde{\Prob}(\widetilde{X}_T = \infty) > 0$ for some $T > 0$.
It is important to note that $\Prob$ and $\widetilde{\Prob}$ are not equivalent if $\widetilde{X}$ is a strict $\Prob$-local martingale or has positive  probability (under $\Prob$) to hit zero.
See again \citet{Ruf_Novikov} for a proof  of these statements.

We now recall the stopped  $\Q$-QNV process $X$, defined as $X := Y^S$, where $S$ denotes the first hitting time of zero by the $\Q$-QNV process $Y$.  In particular, $X$ is a nonnegative $\Q$-local martingale.  
In order to allow for the not necessarily equivalent change to the F\"ollmer measure, we shall assume, throughout the following sections,  that $\Omega$ is the space of nonnegative continuous paths, taking values in $[0,\infty]$ and being absorbed when hitting either zero or infinity, and that $X$ is the canonical process on this space.
Moreover, we shall assume that $\mathcal{F} = \bigvee_{n \in \N} \CF_{\tau_n}$, where $\{\tau_n\}_{n \in \N}$ again denotes the first hitting times of levels $n \in \N$ by ${X}$, and that the
 F\"ollmer measure corresponding to $X$ exists; we shall denote it by  $\widehat{\Q}$. 

 The change of measure from $\Q$ to $\widehat{\Q}$ has the financial interpretation as the change of risk-neutral dynamics if the num\'eraire changes.  For example, if $X$ represents the price of Euros in Dollars and if $\Q$ is a risk-neutral measure for
prices denoted in Dollars, then  $\widehat{\Q}$ corresponds
to the measure under which asset prices denoted in Euros (instead of Dollars) follow local martingale dynamics.

Next, we study this change of measure for stopped QNV processes and observe that the class of stopped QNV processes is stable under changes of num\'eraires, a feature that makes QNV processes attractive as models for foreign
exchange rates.
We start with a simple observation that is related to the statement of Lemma~\ref{L reciprocal}:
\begin{lemma}[Roots of quadratic polynomial]  \label{L roots}
        We consider the polynomial $P(z) := e_1 z^2 + e_2 z + e_3$ of Section~\ref{SS qnv model} and
        its counterpart $\widehat{P}(z) :=-z^2 P(1/z) = -e_3 z^2 -e_2 z - e_1$. They satisfy the following duality relations:
        \begin{itemize}
            \item[(i)] $P$  has only complex roots if and only if $\widehat{P}$ has only complex ones;
            \item[(ii)]   $P$ has two non-zero roots if and only if $\widehat{P}$ has only two non-zero roots;
		\item[(iii)] $P$ has zero as a single root if and only if $\widehat{P}(z)$ is linear and non-constant,
                and vice versa;
            \item[(iv)]  $P$ has zero as a double root if and only if $\widehat{P}$ is constant,
                and vice versa.
        \end{itemize}
\end{lemma}
\begin{proof}
    The statement follows from simple considerations, such as that
    if $r \in \R\setminus\{0\}$ is a root of $P$,
    then $1/r$ is a root of $\widehat{P}$.
\end{proof}
In the context of the next proposition, we remind the reader of \eqref{E hat f}, which we shall utilize in the next section.
\begin{prop}[Closedness under a  change of num\'eraire]  \label{P closedness}
    The process $\widehat{X} := 1/X$ is a stopped QNV process  with polynomial $\widehat{P}(z) := -e_3 z^2 - e_2 z - e_1$ under the F\"ollmer measure $\widehat{\Q}$ corresponding to $X$ as underlying. In particular, the following hold:
    \begin{itemize}
        \item[(i)] $X$ under $\Q$ is a QNV process with complex roots if and only if $\widehat{X}$ is one under $\widehat{\Q}$;
        \item[(ii)] $X$ under $\Q$ is a QNV process with two real non-zero roots  if and only if $\widehat{X}$ is one under $\widehat{\Q}$;
        \item[(iii)] $X$ under $\Q$ is a QNV process with a single root at zero  if and only if $\widehat{X}$  under $\widehat{\Q}$ is a (possibly stopped) shifted geometric Brownian motion;\footnote{We call  a QNV process with linear, but not constant, polynomial ``shifted geometric Brownian motion.''}
        \item[(iv)] $X$ under $\Q$ is a QNV process with a double root at zero if and only if $\widehat{X}$  under $\widehat{\Q}$ is a  (constantly time-changed) stopped Brownian motion.
    \end{itemize}
\end{prop}
\begin{proof}
    The reciprocal $\widehat{X}$ of $X$ is a $\widehat{\Q}$-local martingale by the discussion at the beginning of this section. Let $\{\tau_n\}_{n \in \N}$ again denote  the first hitting times of levels $n \in \N$ by $X$. Then we observe that $\widehat{\Q}$ is absolutely continuous with respect to $\Q$ on $\CF_{\tau_n}$ for all $n \in \N$.  Thus, by Lenglart's extension of Girsanov's theorem \citep[see also Theorem~VIII.1.4 in][]{RY}, $\widehat{X}^{\tau_n}$ is (up to stopping) a QNV process with polynomial $\widehat{P}$ for all $n \in \N$.  Since we assumed that $\bigvee_{n \in \N} \CF_{\tau_n} = \CF$ and, therefore, $\widehat{\Q}$ is uniquely determined through the $\pi$-system $\cup_{n \in \N} \CF_{\tau_n}$, we can conclude that $\widehat{X}$ is a stopped $\widehat{\Q}$-QNV process with polynomial $\widehat{P}$.  
    The statements in (i) to (iv) follow
    from Lemma~\ref{L roots}.
\end{proof}

We now are ready to give a simple proof of Proposition~\ref{P qnv martingality}, stated in the introduction:\\
\emph{Proof of Proposition~\ref{P qnv martingality}.}  
	Consider the probability measure $\widehat{\Q}$ and the stopped $\widehat{\Q}$-QNV process $\widehat{X}$ introduced in Proposition~\ref{P closedness}.  
	By the discussion at the beginning of this section, strict local martingality of $X$ is equivalent to $\widehat{\Q}(\widehat{X}_T =0) > 0$  for some $T>0$.  We shall use the fact, discussed in Section~\ref{SS qnv model}, that $\widehat{X}$ does not hit any  roots of $\widehat{P}$. Thus, $\widehat{\Q}(\widehat{X}_T =0) > 0$ for some $T>0$ is equivalent to 
the polynomial $\widehat{P}$ not having any nonnegative real roots less than or equal to $1/x_0$. This again is equivalent to $e_1 \neq 0$ together with the condition that $P$ has no roots greater than or equal to $x_0$; this is due to the fact that 
$e_1 = 0$ implies that $0$ is a root of $\widehat{P}$ and that $r \in (0, \infty)$ is a root of $P$ if and only if $1/r$ is a root of $\widehat{P}$; see Lemma~\ref{L roots}.  Thus, we have proven the statement concerning the martingality of $X$.

	If $e_1 = 0$, then $Y$ is either constant or Brownian motion (if $e_2 = 0$) or $\widetilde{Y} := Y + e_3/e_2$ is geometric Brownian motion  (if $e_2 \neq 0$) . In all these cases, $Y$ is a true martingale. If $y_0$ lies between two roots
	of $P$, then $Y$ is bounded and thus a martingale. For the reverse direction,  assume that $Y$ is a martingale and that $e_1 \neq 0$. Then there exists a root $r\geq y_0$ of $P$ since otherwise $X = Y^S$ is a strict local martingale. Denote the second root of $P$ by $\widetilde{r}$ and define the QNV process $\widetilde{Y} := r - Y$ with polynomial $P^{\widetilde{Y}} (z)= -e_1 z(z - (r - \widetilde{r}))$. It is clear that  $\widetilde{Y}$ is again a martingale and thus, by the same argument, $r - \widetilde{r} \geq \widetilde{Y}_0 = r-y_0$, which yields the statement.
\qed

Alternatively, we could have used the criterion in \citet{Kotani_2006} to prove Proposition~\ref{P qnv martingality}.

\section{Semistatic hedging}  \label{S semistatic}
In the following, we present an interesting symmetry that can be applied for the semistatic replication of barrier options in certain parameter setups.
\begin{prop}[Symmetry]  \label{P complexsymmetry}
    We fix $T>0$ and assume that $X$ is a stopped QNV process with a polynomial of the form
    $P(z) = e z^2/L + e_2 z+ eL$ and $x_0 = L$ for some $L > 0$ and $e,e_2 \in \R$.
    Let $h:[0,\infty] \rightarrow [0,\infty]$ denote some
    measurable nonnegative function satisfying $h(0) =0$ and $h(\infty) \in \R$.
    We then have the equivalence
    \begin{align*}
        h\left(\frac{X_T}{L}\right) \in \mathcal{L}^1(\Q) \text{  } \Leftrightarrow \text{  }
            h\left(\frac{L}{X_T}\right) \frac{X_T}{L}\in \mathcal{L}^1(\Q),
    \end{align*}
	where $\mathcal{L}(\Q)$ denotes the space of integrable random variables with respect to $\Q$, and the identity
    \begin{align}  \label{E_zr_sym}
        \E\left[h\left(\frac{X_T}{L}\right)\right] = \E\left[h\left(\frac{L}{X_T}\right) \frac{X_T}{L}\right].
    \end{align}
    In particular, by using $h(x) = \1_{x>0} \1_{x < \infty}$, we obtain that
    \begin{align*}
        \E[X_T] = L \Q(X_T > 0),
    \end{align*}
    and by replacing $h(x)$ by $h(x) \1_{x>L}$,
    \begin{align}  \label{E_zr_sym_ind}
        \E\left[h\left(\frac{X_T}{L}\right) \1_{\{X_T > L\}}\right] = \E\left[h\left(\frac{L}{X_T}\right) \frac{X_T}{L} \1_{\left\{X_T < L\right\}}\right].
    \end{align}
\end{prop}
\begin{proof}
    We observe that $Z = \{Z_t\}_{t \geq 0}$, defined by $Z_t := X_t/L$ for all $t \geq 0$, is a stopped QNV process with a polynomial of the form
    $P(z) = e z^2 + e_2 z + e$ and satisfies $Z_0 = 1$. Thus, we can assume, without loss of generality, that $L=1$.
    Now, \eqref{E Qe}, with $H = h(X_T)$, yields, for the F\"ollmer measure $\widehat{\Q}$ corresponding to $X$, that
    \begin{align*}
        \E\left[h\left(X_T\right)\right] =
        \E^{\widehat{\Q}}\left[\frac{h\left(X_T\right)}{X_T}\right]=
        \E\left[h\left(\frac{1}{X_T}\right)X_T\right],
    \end{align*}
    where the second equality follows from observing that $1/X$
    has the same distribution under $\widehat{\Q}$ as $X$ has under $\Q$; see Proposition~\ref{P closedness}.
    This shows \eqref{E_zr_sym}, and the other parts of the
    statement follow  from it directly.
\end{proof}

\begin{remark}[Alternative proof of Proposition~\ref{P complexsymmetry}]
    Proposition~\ref{P complexsymmetry} can also directly be shown without relying
    on the F\"ollmer measure.  For this, we again assume $L=1$ and define the sequences of processes
    $X^n$ ($X^{1/n}$) by stopping $X$ as soon as it hits $n$ ($1/n$) for all $n \in
    \N$. We then observe that Girsanov's theorem  \citep[Theorem~VIII.1.4 of][]{RY} implies for all $\epsilon > 0$ and
    all Borel sets $A_\epsilon
    \subset (\epsilon, \infty)$
    \begin{align*}
        \E\left[X^n_T \1_{\{1/X^n_T \in A_\epsilon\}}\right] =
            \Q\left(X_{T}^{1/n} \in A_\epsilon\right).
    \end{align*}
    Now, we first let $n$ go to $\infty$ and then $\epsilon$ to
    zero and obtain that
    \begin{align*}
    \E\left[X_T \1_{\{1/X_T \in A\}}\right] =
            \Q\left(X_{T}\in A \cap (0,\infty)\right)
    \end{align*}
    for all Borel sets $A$, which again yields  \eqref{E_zr_sym}. \qed
\end{remark}


\begin{remark}[Semistatic hedging]
    Proposition~\ref{P complexsymmetry} and, in particular,
    \eqref{E_zr_sym_ind} can be interpreted as the existence of a semistatic hedging
    strategy for barrier options in the spirit of
    \citet{Bowie_Carr}, \citet{Carr_Ellis_Gupta}, and \citet{Carr_Lee_2009}.

    To see this, consider
    a QNV process $X$ with a polynomial of the form
    $P(z) = ez^2/L + e_2 z+ e L$ and $x_0 > L$ for some $L > 0$ and $e,e_2 \in \R$. Consider further
    a down-and-in barrier option with barrier $L$ and terminal payoff $h(X_T/L)$ if
    the barrier is hit by $X$. For a semistatic hedge, at time zero one buys
    two positions of European claims, the first paying off
    $h\left({X_T}/{L}\right) \1_{\{X_T \leq L\}}$ and the second paying off
    $h\left({L}/{X_T}\right) {X_T}/{L} \1_{\left\{X_T < L\right\}}$.  If the barrier
    is not hit, both positions have zero price at time $T$. If the barrier is
    hit, however, one sells the second position and buys instead a third position paying
    off $h\left({X_T}/{L}\right) \1_{\{X_T > L\}}$. The equality in \eqref{E_zr_sym_ind}
    guarantees that these two positions have the same price at the hitting time of the barrier.
    This strategy is semistatic
    as it requires trading only at maximally two points of time.

    Proposition~\ref{P complexsymmetry}, in particular, contains the well-known case of
    geometric Brownian motion ($e=0$), where semistatic hedging is always possible.
    It is an open question to determine more general symmetries
    than that of Proposition~\ref{P complexsymmetry}. One difficulty  arises here
    from the lack of equivalence of the two measures $\Q$ and $\widehat{\Q}$. However, adding an independent
    change of time to the dynamics of $X$ preserves any existing such symmetry.
\qed
\end{remark}

\section{Joint replication and hyperinflation}     \label{S hyperinflation}
In this section, we continue with a financial point of view and interpret the probability measure $\Q$ as the unique risk-neutral measure, under which the stopped QNV process $X$ denotes the price of an asset,
say, the price of a Euro in Dollars. The F\"ollmer measure $\widehat{\Q}$, introduced in Section~\ref{S closedness},
can then be interpreted as the unique risk-neutral probability measure of a European investor who uses the price of a Euro as a num\'eraire. To emphasize this point further, we shall use  the notation $\Qd := \Q$ and $\Qe := \widehat{\Q}$ from now on.

Throughout this section, we assume a finite time horizon $T<\infty$. In the spirit of \citet{CFR2011}, we describe a contingent claim by a pair $D = (D^\$, D^\eu)$ of random variables such that $D^\$ = H^\$(\{X_t\}_{t \in [0,T]})$
and $D^\eu = H^\eu(\{X_t\}_{t \in [0,T]})$ for two measurable functions $H^\$, H^\eu: C([0,T],[0,\infty]) \rightarrow [0,\infty]$ satisfying $D^\eu = D^\$ / X_T$ on the event $\{0 < X_T < \infty\}$. The first component of $D$ represents the claim's (random) payoff denoted in Dollars at time $T$ as seen under the measure $\Q^\$$  and the second its (random) payoff denoted in Euros at time $T$ as seen under the measure $\Q^\eu$. In particular, the condition  $D^\eu = D^\$ / X_T$ on the event $\{0 < X_T < \infty\}$ ensures that both the American and the European investors receive the same payoff (in different currencies) in the states of the world that both measures $\Qd$ and $\Qe$ can ``see.''  

We are interested in the quantity     
\begin{align}  \label{E ppdd}
        p^{\$}(D) &= \E^{\Qd}\left[D^{\$}\right] + x_0 \E^{\Qe}\left[D^{\eu} \1_{\{1/X_T = 0\}}\right], 
\end{align}
derived in \citet{CFR2011}.  This quantity describes a possible price (in Dollars), the \emph{minimal joint replicating price}, for the claim $D$.  Under additional assumptions on the completeness of the underlying market, this price can be interpreted as the minimal cost of superreplicating the claim's payoffs almost surely under both measures $\Qd$ and $\Qe$; observe that these two probability measures are not equivalent if $X$ is a strict local martingale.

\begin{cor}[Minimal joint replicating price in a QNV model]  \label{C minimal super}
    We have the identity
    \begin{align}  
        p^\$(D) =&\,\, \E\left[H^{\$}\left(\{f(W_t^S)\}\tInd\right)  \1_{\{\tau > S \wedge T\}}
            \exp\left(\frac{C(T \wedge S)}{2}\right) g^\$(W_T^{S})\right]   \label{E p qnv}\\ &+
                x_0 \E\left[H^{\eu}\left(\{f(W_t^\tau)\}\tInd\right) \1_{\{\tau \leq S \wedge T\}} \exp\left(\frac{C(T\wedge \tau)}{2}\right)
                 g^{\eu}(W_T^{\tau})\right],  \nonumber
    \end{align}
    where $W = \{W_t\}_{t\geq 0}$ denotes a Brownian motion and $S$ ($\tau$) the first hitting time of zero (infinity) by $\{f(W_t)\}_{t \geq 0}$ and $C, f$, and $g^\$ \equiv g$ are as in Corollary~\ref{C computation}; moreover, we use  $g^\eu$ similarly but corresponding to the stopped QNV process $\widehat{X}$ of Proposition~\ref{P closedness}.
\end{cor}
\begin{proof}
	By Corollary~\ref{C computation}, the first term on the right-hand side of \eqref{E p qnv} corresponds to $\E^{\Qd}\left[D^{\$}\right] $. Now, Proposition~\ref{P closedness} yields that $\widehat{X} = 1/X$ is a stopped 
	$\Qe$-QNV process for some polynomial $\widehat{P}$.  With the representation $\widehat{f}$ of $\widehat{X}$ (along with $\widehat{C}$, among others) in Theorem~\ref{T QNV} and Corollary~\ref{C computation},  Lemma~\ref{L reciprocal} yields that $\widehat{f} = 1/f$, that $\widehat{C} = C$, and that $S$ and $\tau$ interchange places.  Since $\{\widehat{f}(W^\tau_T) = 0\} \cap \{S>\tau \wedge T\} = \{\tau \leq T\}\cap \{\tau \wedge T<S\} =\{\tau \leq S \wedge T\}$, we conclude by applying Corollary~\ref{C computation} to the second term on the right-hand side of \eqref{E ppdd}.
\end{proof}

We emphasize the symmetry of $\tau$ and $S$, which we relied on in the proof of the corollary. The stopping time $\tau$ is the first time $f(W)$ hits infinity and $\widehat{f}(W) = 1/f(W)$ hits zero and  the stopping time
$S$ satisfies the converse statement. 

We also remark that the probability measure $\Prob$ of Theorem~\ref{T QNV}  can be interpreted as a physical measure, under which hyperinflations occur with positive probability. Thus, $f(W^S)$ can be used to model an exchange rate that allows (under $\Prob$)
for hyperinflations in either Euros or Dollars; then $p^\$$ represents the minimal replicating cost (in Dollars) for a claim that pays $D^\$$ Dollars if no hyperinflation of the Dollar occurs and that pays $D^\eu$ Euros if the Dollar hyperinflates, 
corresponding to the Dollar price of a Euro being infinity. For a more thorough discussion on this interpretation, we refer the reader to \citet{CFR2011}.

	Lemma~\ref{L reciprocal} shows that $g^{\eu}(W_T^\tau) = g^{\$}(W_T^\tau) f(W_T^\tau)/x_0$ on $\{\tau \leq S\}$, with multiplications of zero and infinity formally interpreted in such a way as to obtain equality.  Moreover,
after setting  $D^\eu = \infty$ on $\{S \leq \tau \wedge T\}$ and $D^\$= \infty$ on $\{\tau \leq S \wedge T\}$, we also have that $D^{\eu} = D^\$/f(W_T^{\tau \wedge S})$, again with undefined expressions interpreted in such a way to obtain equality. Therefore,
the expression in \eqref{E p qnv} formally reduces to
    \begin{align}  \label{E p qnv_symbolic}
        p^\$(D) = \text{``}\E^\Prob\left[H^{\$}\left(\{f(W_t^{\tau \wedge S})\}\tInd\right)  \exp\left(\frac{C(T \wedge S \wedge \tau)}{2}\right)  g^\$(W_T^{\tau \wedge S})\right],\text{''}
    \end{align}
    where  all multiplications of $0$ with $\infty$ are interpreted in the sense of \eqref{E p qnv}. Indeed, if all multiplications are well defined, \eqref{E p qnv_symbolic} exactly corresponds to \eqref{E p qnv}.

As a brief illustration of the last corollary, consider the minimal joint replicating price of one Euro in Dollars,
to wit, $D=(X_T,1)$. From \eqref{E p qnv}  we obtain that
    \begin{align*} 
        p^\$(D) &= \E\left[\exp\left(\frac{C(T \wedge S \wedge \tau)}{2}\right) \left( \left(f(W^{\tau \wedge S}_T) \1_{\{\tau > S \wedge T\}}\right)
             g^\$(W_T^{\tau \wedge S})  +
                x_0 \1_{\{\tau \leq S \wedge T\}} 
                 g^{\eu}(W_T^{\tau \wedge S})\right)\right]\\
		&= x_0 \E\left[\exp\left(\frac{C(T \wedge S \wedge \tau)}{2}\right) \left(\1_{\{\tau > S \wedge T\}}
             g^\eu(W_T^{\tau \wedge S})  +\1_{\{\tau \leq S \wedge T\}} 
                 g^{\eu}(W_T^{\tau \wedge S})\right)\right]\\
		&= x_0 \E\left[\exp\left(\frac{C(T \wedge S \wedge \tau)}{2}\right)
             g^\eu(W_T^{\tau \wedge S}) \right] = x_0,
    \end{align*}
where we have used the representation $g^\eu = f g^\$/x_0$ of Lemma~\ref{L reciprocal}.  Thus, the minimal replicating cost for one Euro is $x_0$ Dollars, exactly what we hoped for. We remark that the symbolic representation of \eqref{E p qnv_symbolic} directly yields the same statement, too.

\appendix
\section{A technical result}  \label{A technical}
The following lemma is used in the proof of Theorem~\ref{T QNV}:
\begin{lemma}[Necessary condition for path independence of
integrals] \label{L pathindepence}
    Let $W = \{W_t\}_{t \geq 0}$ denote a Brownian motion and 
    $\tau$  a stopping time of the form \eqref{E t1 tau}
    for some $a,b \in [-\infty, \infty]$ with $a <0 < b$.
    Let $h:(a,b) \rightarrow \R$ denote a continuous function and assume that
    \begin{align}  \label{E hg}
        \int_0^{t} h(W_s) \dd s = \widetilde{h}(t,W_t)
    \end{align}
    almost surely on $\{\tau > t\}$ for all $t \geq 0$, for some measurable function $\widetilde{h}: [0, \infty) \times (a,b) \rightarrow \R$.
    Then $h(\cdot) \equiv C$ for
    some $C\in \R$.
\end{lemma}
\begin{proof}
    Assume that \eqref{E hg} holds but $h$ is not a constant.
    Then there exist some
    $\epsilon > 0$ and some $y \in (a+\epsilon,b-\epsilon)$ such that $|h(y) - h(0)| =
    5 \epsilon$; without loss of generality, assume $y \in (0,b-\epsilon)$. Now, define
    $\tilde{y} := \inf\{y \in [0,b)| |h(y) - h(0)| \geq
   5\epsilon\}$. Assume, again without loss of generality, that $h(0) = 0$ and
    $h(\tilde{y}) =5 \epsilon> 0.$ Observe that there exists some $\delta \in (0, \min(-a,\epsilon))$
    such that $h(y) < \epsilon$ for all $y$ with $|y| <
    \delta$ and $h(y) \geq 4 \epsilon$ for all $y$
    with $|y - \tilde{y}| < \delta.$  In summary, we assume for the rest of the proof that $h(y) > -5\epsilon$ for all $y \in (-\delta, \tilde{y} + \delta)$, that $h(y) \leq \epsilon$ for all $y \in (-\delta, \delta)$, and that $h(y) \geq 4\epsilon$ for all $y \in (\tilde{y}-\delta, \tilde{y}+\delta)$.

    Now, fix $T>0$ and define
	$$\mathfrak{A} := \left\{\omega \in C([0,T], \R) \left| -\delta < \inf_{s \in [0,T]} \omega_s \leq \sup_{s \in [0,T]} \omega_s < \tilde{y}+\delta, \int_0^{T} h(\omega_s) \dd s = \widetilde{h}(T,\omega_T)\right.\right\}$$
 	and $\mathfrak{B} := \{\omega_T: \omega \in \mathfrak{A}\}$.
    Since $h$ is uniformly continuous on $[-\delta, \tilde{y}+\delta]$, the mapping $\mathfrak{A} \ni \omega \mapsto  \int_0^{T} h(\omega_s) \dd s = \widetilde{h}(T, \omega_T)$ is continuous if $\mathfrak{A}$ is equipped with the supremum norm. Therefore, $\widetilde{h}(T,\cdot): \mathfrak{B} \rightarrow \R$ is continuous and there exists $\widetilde{\delta} \in (0, \delta)$ such that $|\widetilde{h}(T,y_1) - \widetilde{h}(T,y_2)| < \epsilon T$ for all $y_1,y_2 \in (-\widetilde{\delta}, \widetilde{\delta}) \cap \mathfrak{B}$ since  \eqref{E hg} almost surely implies that $\mathfrak{B}$ is dense in $(-\delta, \tilde{y}+\delta)$. 

	There exists $\omega \in \mathfrak{A}$ such that  $-\widetilde{\delta} < \inf_{s \in [0,T]} \omega_s \leq \sup_{s \in [0,T]} \omega_s < \widetilde{\delta}$ since \eqref{E hg} holds almost surely; thus $\widetilde{h}(T,y) < 2 \epsilon T$ for all  $y \in (-\widetilde{\delta}, \widetilde{\delta}) \cap \mathfrak{B}$.  In order to obtain a contradiction we now consider an $\omega \in \mathfrak{A}$ such that $\omega_T \in (-\widetilde{\delta}, \widetilde{\delta})$ and $\widetilde{h}(T,\omega_T) > 2 \epsilon T$.  Towards this end, choose an $\omega \in \mathfrak{A}$ such that  $\tilde{y} - \delta < \inf_{s \in [0.1T,0.9T]} \omega_s \leq \sup_{s \in [0.1T,0.9T]} \omega_s < \tilde{y} + \delta$ and  $\omega_T \in (-\widetilde{\delta}, \widetilde{\delta})$, which again always exists. Observe that for this choice of $\omega$ we have $\widetilde{h}(T,\omega_T) \geq -0.2 T \cdot 5 \epsilon + 0.8 T \cdot 4 \epsilon > 2 \epsilon T$.
 Thus, under the assumptions of the lemma, $h$ is constant.
\end{proof}

\bibliographystyle{apalike}
\setlength{\bibsep}{1pt}
\small\bibliography{aa_bib}{}

\begin{thebibliography}{}

\bibitem[Albanese et~al., 2001]{Albanese_BS_hyper}
Albanese, C., Campolieti, G., Carr, P., and Lipton, A. (2001).
\newblock Black-{S}choles goes hypergeometric.
\newblock {\em Risk Magazine}, 14(12):99--103.

\bibitem[Andersen, 2011]{Andersen}
Andersen, L. (2011).
\newblock {Option pricing with quadratic volatility: a revisit}.
\newblock {\em Finance and Stochastics}, 15(2):191--219.

\bibitem[Bluman, 1980]{Bluman_transformation}
Bluman, G. (1980).
\newblock On the transformation of diffusion processes into the {W}iener
  process.
\newblock {\em {SIAM} Journal on Applied Mathematics}, 39(2):238--247.

\bibitem[Bluman, 1983]{Bluman_mapping}
Bluman, G. (1983).
\newblock On mapping linear partial differential equations to constant
  coefficient equations.
\newblock {\em {SIAM} Journal on Applied Mathematics}, 43(6):1259--1273.

\bibitem[Bowie and Carr, 1994]{Bowie_Carr}
Bowie, J. and Carr, P. (1994).
\newblock Static simplicity.
\newblock {\em Risk Magazine}, 7(8):45--50.

\bibitem[Carr et~al., 1998]{Carr_Ellis_Gupta}
Carr, P., Ellis, K., and Gupta, V. (1998).
\newblock Static hedging of exotic options.
\newblock {\em Journal of Finance}, 53(3):1165--1190.

\bibitem[Carr et~al., 2013]{CFR2011}
Carr, P., Fisher, T., and Ruf, J. (2013).
\newblock On the hedging of options on exploding exchange rates.
\newblock Preprint, arXiv:1202.6188.

\bibitem[Carr and Lee, 2009]{Carr_Lee_2009}
Carr, P. and Lee, R. (2009).
\newblock Put-call symmetry: extensions and applications.
\newblock {\em Mathematical Finance}, 19(4):523--560.

\bibitem[Carr et~al., 2002]{CLM_reduction}
Carr, P., Lipton, A., and Madan, D. (2002).
\newblock The reduction method for valuing derivative securities.

\bibitem[Chibane, 2011]{Chibane}
Chibane, M. (2011).
\newblock Analytical option pricing for time dependent quadratic local
  volatility models.

\bibitem[Delbaen and Schachermayer, 1995]{DS_Bessel}
Delbaen, F. and Schachermayer, W. (1995).
\newblock Arbitrage possibilities in {Bessel} processes and their relations to
  local martingales.
\newblock {\em Probability Theory and Related Fields}, 102(3):357--366.

\bibitem[F\"ollmer, 1972]{F1972}
F\"ollmer, H. (1972).
\newblock The exit measure of a supermartingale.
\newblock {\em Zeitschrift f{\"u}r Wahrscheinlichkeitstheorie und Verwandte
  Gebiete}, 21:154--166.

\bibitem[Goldys, 1997]{Goldys_1997}
Goldys, B. (1997).
\newblock A note on pricing interest rate derivatives when forward {LIBOR}
  rates are lognormal.
\newblock {\em Finance and Stochastics}, 1:345--352.

\bibitem[Hirsch and Smale, 1974]{Hirsch}
Hirsch, M.~W. and Smale, S. (1974).
\newblock {\em Differential Equations, Dynamical Systems, and Linear Algebra}.
\newblock Academic Press, New York.

\bibitem[Ingersoll, 1997]{Ingersoll_bounded}
Ingersoll, J.~E. (1997).
\newblock Valuing foreign exchange rate derivatives with a bounded exchange
  process.
\newblock {\em Review of Derivatives Research}, 1:159--181.

\bibitem[Jacod and Shiryaev, 2003]{JacodS}
Jacod, J. and Shiryaev, A.~N. (2003).
\newblock {\em Limit Theorems for Stochastic Processes}.
\newblock Springer, Berlin, 2nd edition.

\bibitem[Karatzas and Shreve, 1991]{KS1}
Karatzas, I. and Shreve, S.~E. (1991).
\newblock {\em Brownian Motion and Stochastic Calculus}.
\newblock Springer, New York, 2nd edition.

\bibitem[Kotani, 2006]{Kotani_2006}
Kotani, S. (2006).
\newblock On a condition that one-dimensional diffusion processes are
  martingales.
\newblock In {\em In Memoriam Paul-Andr\'e Meyer: S\'{e}minaire de
  Probabilit\'{e}s, XXXIX}, pages 149--156. Springer, Berlin.

\bibitem[Lipton, 2001]{Lipton_book}
Lipton, A. (2001).
\newblock {\em Mathematical Methods for Foreign Exchange: a Financial
  Engineer's Approach}.
\newblock World Scientific.

\bibitem[Lipton, 2002]{Lipton_volsmile}
Lipton, A. (2002).
\newblock The vol smile problem.
\newblock {\em Risk Magazine}, 15(2):61--65.

\bibitem[McKean, 1969]{McKean_1969}
McKean, H.~P. (1969).
\newblock {\em {Stochastic Integrals}}.
\newblock {Academic Press}, New York.

\bibitem[Meyer, 1972]{M}
Meyer, P. (1972).
\newblock La mesure de {H. F\"ollmer} en th\'{e}orie de surmartingales.
\newblock In {\em S\'{e}minaire de Probabilit\'{e}s, VI}, pages 118--129.
  Springer, Berlin.

\bibitem[Miltersen et~al., 1997]{Miltersen_1997}
Miltersen, K.~R., Sandmann, K., and Sondermann, D. (1997).
\newblock Closed form solutions for term structure derivatives with log-normal
  interest rates.
\newblock {\em Journal of Finance}, 52(1):409--430.

\bibitem[Perkowski and Ruf, 2012]{Perkowski_Ruf}
Perkowski, N. and Ruf, J. (2012).
\newblock Conditioned martingales.
\newblock {\em Electronic Communications in Probability}, 17(48):1--12.

\bibitem[Protter, 2003]{Protter}
Protter, P.~E. (2003).
\newblock {\em Stochastic Integration and Differential Equations}.
\newblock Springer, New York, 2nd edition.

\bibitem[Rady, 1997]{Rady}
Rady, S. (1997).
\newblock Option pricing in the presence of natural boundaries and a quadratic
  diffusion term.
\newblock {\em Finance and Stochastics}, 1(4):331--344.

\bibitem[Rady and Sandmann, 1994]{Rady_Sandmann}
Rady, S. and Sandmann, K. (1994).
\newblock The direct approach to debt option pricing.
\newblock {\em Review of Futures Markets}, 13(2):461--515.

\bibitem[Revuz and Yor, 1999]{RY}
Revuz, D. and Yor, M. (1999).
\newblock {\em Continuous Martingales and Brownian Motion}.
\newblock Springer, Berlin, 3rd edition.

\bibitem[Ruf, 2013]{Ruf_Novikov}
Ruf, J. (2013).
\newblock A new proof for the conditions of {N}ovikov and {K}azamaki.
\newblock {\em Stochastic Processes and Their Applications}, 123:404--421.

\bibitem[Z\"uhlsdorff, 2001]{Zuehlsdorff_2001}
Z\"uhlsdorff, C. (2001).
\newblock The pricing of derivatives on assets with quadratic volatiltiy.
\newblock {\em Applied Mathematical Finance}, 8(4):235--262.

\bibitem[Z\"uhlsdorff, 2002]{Zuehlsdorff_2002}
Z\"uhlsdorff, C. (2002).
\newblock Extended {L}ibor market models with affine and quadratic volatiliy.

\end{thebibliography}
\end{document}